\documentclass{ifacconf}
\usepackage{graphicx} 
\usepackage{amsmath}
\usepackage{amsfonts}
\usepackage{mathtools}
\usepackage{dirtytalk}
\usepackage{xcolor}
\usepackage{arydshln}
\usepackage{bm}
\usepackage{natbib}

\usepackage{amssymb}

\newtheorem{lemma}{Lemma}
\newtheorem{theorem}{Theorem}

\newtheorem{remark}{Remark}
\newtheorem{definition}{Definition}

\newtheorem{proof}{Proof}

\AtEndEnvironment{problem}{\hfill$\Box$}
\AtEndEnvironment{proposition}{\hfill$\Box$}
\AtEndEnvironment{lemma}{\hfill$\Box$}
\AtEndEnvironment{theorem}{\hfill$\Box$}
\AtEndEnvironment{assumption}{\hfill$\Box$}
\AtEndEnvironment{condition}{\hfill$\Box$}
\AtEndEnvironment{remark}{\hfill$\Box$}
\AtEndEnvironment{corollary}{\hfill$\Box$}
\AtEndEnvironment{definition}{\hfill$\Box$}
\AtEndEnvironment{proof}{\hfill$\Box$}

\newcommand{\norm}[1]{\left|#1\right|}

\DeclareMathOperator{\R}{\mathbb{R}}

\providecommand{\SO}{\mathbf{SO}}
\providecommand{\SE}{\mathbf{SE}_2(3)}

\usepackage{xcolor}

\begin{document}
\begin{frontmatter}
\title{Relative Pose-Velocity Estimation Using Dual IMU Measurements and Relative Position Sensing\thanksref{footnoteinfo}}
\thanks[footnoteinfo]{This work was supported by the"Grands Fonds Marins" Project Deep-C, and the ASTRID ANR project ASCAR ANR-23-ASTR-0016. This research work is also supported in part by NSERC-DG RGPIN-2020-04759 and Fonds de recherche du Québec (FRQ).}

\author[First]{Alessandro Melis} 
\author[First]{Tarek Bouazza} 
\author[Second]{Soulaimane Berkane} 
\author[First,Third]{Tarek Hamel}

\address[First]{I3S-CNRS, Nice-Sophia Antipolis, France \\(email: melis@i3s.unice.fr, bouazza@i3s.unice.fr).}
\address[Second]{Department of Computer Science and Engineering, Université du Québec en Outaouais (UQO), QC J8X 3X7, Canada (email:soulaimane.berkane@uqo.ca).}
\address[Third]{Institut Universitaire de France, Nice-Sophia Antipolis, France (email:thamel@i3s.unice.fr).}
\begin{abstract}   
This paper addresses the problem of estimating the relative pose (position and orientation) and velocity of a vehicle with respect to a moving target, where both are equipped with Inertial Measurement Units (IMUs), assuming the availability of relative position or bearing measurements. The body-target relative dynamics are formulated on $\SE$ and recast into a linear time-varying (LTV) model in the ambient space $\mathbb{R}^{15}$, on which a deterministic Riccati observer is designed. We analyze the uniform observability (UO) conditions required to guarantee global exponential convergence of the estimation error in the ambient space for both measurement cases. In the case of relative position measurements, UO requires only a persistence-of-excitation condition on the target acceleration, whereas for bearing measurements, additional conditions are required. Building on this, a nonlinear complementary filter on $\SO(3)$ is designed to provide a smooth estimate of the orientation component of the state with almost global asymptotic stability. Finally, simulation results are provided to validate the proposed solution.
\end{abstract}

\end{frontmatter}
\section{Introduction}
The ability to estimate relative pose (position and orientation) is essential for feedback control in applications involving interacting physical systems. In this context, observability analysis plays a crucial role: it determines whether relative states can be reconstructed from available measurements and how robustly this can be achieved in the presence of noise and model uncertainties. While several works have addressed uniform observability for static targets (\cite{Gintrand2022}, \cite{hamel2017position}), the case of moving targets remains largely open, even though relative state estimation is central to many practical applications. Examples include autonomous landing, docking, and formation control, where control objectives are naturally expressed in relative coordinates.

Autonomous landing of UAVs on moving platforms has been extensively studied using Kalman/EKF-based fusion of IMU and vision/GNSS data \cite{Borowczyk2017, Falanga2017, Araar2017}. These contributions, however, primarily focus on guidance and control, without providing a detailed observability analysis of the estimation model. In contrast, the aerospace community has placed more emphasis on observability in relative pose estimation. Observability conditions have been derived for relative pose estimation between satellites, considering orbital and attitude dynamics \cite{Li2019}, and, in a constructive manner, between a satellite and the rotating Earth \cite{Reis2019}. Nonetheless, the existing results are either local in nature or applicable only to a very restricted class of systems.

In this work, we consider a dual-IMU body-target setup and formulate the relative dynamics on the Lie group $\SE$, which compactly captures orientation, linear velocity, and position. The available measurements include either the relative position of the target origin expressed in the body frame (e.g., acquired with a stereo camera) or, when unavailable, the relative bearing (e.g., when the target is far or the body is equipped with a monocular camera). Following \cite{Benahmed2025}, the relative dynamics are embedded in the ambient space $\mathbb{R}^{15}$ and recast as a linear time-varying (LTV) system. This formulation, firstly applied to the attitude estimation problem in \cite{Batista2012Single,Batista2012Cascade}, overcomes the topological limitations of $\SO(3)$ and enables the design of a linear Riccati observer for the model, thereby allowing the study of estimation-error convergence via the classical LTV uniform observability.

When relative position measurements are available, we show that persistence of excitation of the target acceleration is sufficient to guarantee uniform observability. In the bearing-only case, the time-varying output matrix can couple the bearing excitation with the target acceleration, leading to potential loss of observability along degenerate trajectories where these signals are not independent. To overcome this issue, we introduce persistence of excitation conditions that preclude such degeneracies.

The contributions of this work are threefold: (i) formulation of the dual-IMU body-target relative dynamics on $\SE$ and its embedding into $\mathbb{R}^{15}$, (ii) design of a linear Riccati observer with uniform observability analysis for both relative position and bearing measurements, and (iii) design of a nonlinear complementary filter on $\SO(3)$ providing smooth orientation estimates with almost global asymptotic stability. Simulations validate the proposed approach and illustrate the estimator's convergence properties under realistic conditions.

The rest of this paper is organized as follows: in Section~\ref{section:preliminary material}, we provide the notation and uniform observability definition relevant for the proposed observer stability. Section~\ref{section:problem description} introduces the body-target relative kinematics and describes the measurements used for observer design. Section~\ref{section:observer design} features the proposed observer and uniform observability conditions for both the cases of relative position and bearing measurements. Section~\ref{section:simulations} shows simulation results. Finally, concluding remarks are reported in Section~\ref{section:conclusions}.

\section{Preliminary Material}\label{section:preliminary material}
\subsection{Notation}
The Euclidean norm of a vector $x \in \mathbb{R}^n$ is denoted by $|x|$, the 2-norm and Frobenius norm of a matrix $X\in\R^{n\times n}$ are denoted, respectively by $\|X\|$ and $\|X\|_F$.

    For any vector $\Omega:=\begin{bmatrix}
        \Omega_1 & \Omega_2 & \Omega_3
    \end{bmatrix}^\top \in \mathbb{R}^3$, the skew-symmetric matrix associated with the cross product is defined as 
    \[ 
    \Omega_\times =
    \begin{bmatrix}
        0 & -\Omega_3 & \Omega_2 \\
        \Omega_3 & 0 & -\Omega_1 \\
        -\Omega_2 & \Omega_1 & 0
    \end{bmatrix}.
    \]
    
    The unit n-sphere is represented as $S^n := \{v \in \mathbb{R}^{n+1} \mid |v| = 1\}$.

    We denote the set of $n \times n$ positive definite matrices is denoted by $\mathbb{S}_+(n)$, and the identity matrix is denoted by $I_n \in \mathrm{R}^{n \times n}$
    
    The special orthogonal group, denoted as $\SO(3)$, represents the Lie group of 3D rotations and is given by
    \[
    \SO(3) := \{R \in \mathbb{R}^{3\times3} \mid R^{\top} R = RR^\top = I_3, \det(R) = 1\}.
    \]
    Its associated Lie algebra is defined as
    \[
    \mathfrak{so}(3) := \{\Omega_\times \in \mathbb{R}^{3\times3} \mid \Omega \in \mathbb{R}^3\}.
    \]
    The extended special euclidean group, denoted as $\SE$, represents the Lie group of 3D rigid motions that includes rotations, translations and linear velocities, and is given by
    \begin{equation*}
        \SE:=\left\{\left[
            \begin{array}{ccc}
                R & v & x \\
                0_{2\times3} & \multicolumn{2}{c}{I_2}
            \end{array}
        \right]\in\R^{5\times 5}|R\in\SO(3),p,v\in\R^3\right\}.
    \end{equation*}
    Its associated Lie algebra is defined as
    \begin{equation*}
        \mathfrak{se}_2(3) := \left\{\left[
            \begin{array}{ccc}
                \Omega_\times & a & v \\
                0_{2\times3} & \multicolumn{2}{c}{0_{2\times2}}
            \end{array}
        \right]\in\R^{5\times 5}|R\in\mathfrak{so}(3),p,v\in\R^3\right\}.
    \end{equation*}
    For any $y\in S^{n-1}$, the operator $\pi_y:=I_{n}-yy^\top$ projects any vector $x\in\R^n$ to the plane orthogonal to $y$.
    The Kronecker product between two matrices $A$ and $B$ is denoted by $A\otimes B$.
    The vectorization operator vec$:\R^{m\times n}\to\R^{mn}$ stacks the columns of a matrix $A\in\R^{m\times n}$ into a single column vector in $\R^{mn}$.

\subsection{Observability Principles and Requirements}
Consider a generic linear time-varying (LTV) system:
\begin{subequations} \label{eq:LTV_system_dynamics}
    \begin{align}
        \dot{\bm{x}} &= A(t) \bm{x} + B(t) \bm{u} \label{eq:state_equation}, \\
       \bm{y} &=C(t) \bm{x} \label{eq:measurement_equation},
    \end{align}
\end{subequations}
where \( \bm{x} \in \mathbb{R}^n \) denotes the system state, \( \bm{u} \in \mathbb{R}^s \) the system input, and \( \bm{y} \in \mathbb{R}^m \) the system output. The following definition of observability related to this system is borrowed from the works \cite{aeyels1998asymptotic,chen1984linear}.
\begin{definition}[Uniform Observability] \label{def:uniform_obs}
The system \eqref{eq:LTV_system_dynamics} is said to be \textit{uniformly observable} if there exist constants $\delta, \mu > 0$ such that, for all $t \geq 0$:
\begin{equation}  
W(t, t + \delta) \succeq \mu I_n \succ 0, 
\label{eq:main Uniform Observability condition}
\end{equation}  
where  
\begin{equation}  
W(t, t + \delta) \coloneqq \frac{1}{\delta} \int_{t}^{t+\delta} \Phi^\top(s, t)C^\top(s)C(s)\Phi(s, t) \, ds \notag 
\end{equation} is the observability Gramian, and  
$\Phi(s, t)$ denotes the state transition matrix associated with $A(t)$, defined by: $\frac{d}{dt} \Phi(s, t) = A(t) \Phi(s, t)$, $\Phi(t, t) = I_n$.

If condition \eqref{eq:main Uniform Observability condition} holds, the pair $(A(t), C(t))$ is said to be \textit{uniformly observable}.
\end{definition}

\section{Problem description}\label{section:problem description}
We consider the problem of estimating the relative pose and linear velocity of a body object with respect to (w.r.t.) a, possibly moving, target object.
We assume both the body and the target are equipped with inertial measurement units (IMUs) and provide linear acceleration and angular velocity measurements.

Let $\{\mathcal{B}\}, \{\mathcal{T}\}$ and $\{\mathcal{I}\}$ denote, respectively, the body reference frame, the target reference frame and an inertial reference frame. Denote as $\{e_1,e_2,e_3\}$ the standard basis vectors of $\mathcal{I}$. The orientation of frame $\{\mathcal{B}\}$ (resp. $\{\mathcal{T}\}$) w.r.t. frame $\{\mathcal{I}\}$ is represented by the rotation matrix $Q_B \in \mathbf{SO}(3)$ (resp. $Q_T \in \mathbf{SO}(3)$). 

Denote as $p_B\in \mathbb{R}^3$ (resp. $p_T\in \mathbb{R}^3$), $v_B\in\mathbb{R}^3$ (resp. $v_T\in \mathbb{R}^3$) and $a_B\in \mathbb{R}^3$ (resp. $a_T \in \mathbb{R}^3$) the position, linear velocity and acceleration of the body (resp. target) expressed in frame $\mathcal{I}$.
Let $\Omega\in \mathbb{R}^3$ (resp. $\omega_T\in \mathbb{R}^3$) represent the angular velocity of the body (resp. target) expressed in frame $\mathcal{B}$ (resp. $\mathcal{T}$).

The body and target dynamics are given by the equations
\begin{equation}\label{eq: body target dynamics}
    \begin{cases}
    \dot p_B = v_B\\
    \dot v_B = Q_B a_B + g e_3\\
    \dot{Q}_B = Q_B \Omega_{\times}, 
    \end{cases} \qquad \begin{cases}
    \dot p_T = v_T\\
    \dot v_T = Q_T a_T + g e_3\\
    \dot{Q}_T = Q_T \omega_{T_\times}, 
    \end{cases}
\end{equation}
where $g$ denotes the gravitational acceleration, $a_B$ (resp. $a_T$) denotes the specific acceleration of the body (resp. the target), and 
$\Omega$ (resp. $\omega_T$) denotes the angular velocity of the body (resp. the target), 
all provided by their respective IMUs.

We now adopt the $\SE$ group framework to compactly describe the relative motion dynamics from $\mathcal{T}$ to $\mathcal{B}$. The extended 3D rigid motions from $\mathcal{B}$ to $\mathcal{I}$ and from $\mathcal{T}$ to $\mathcal{I}$, can also be represented as elements of $\SE$ as $X_B,X_T\in\SE$, defined by
\begin{equation}
    X_B := \left[
            \begin{array}{c;{2pt/2pt}cc}
                Q_B & v_B & p_B \\
                \hdashline[2pt/2pt]
                0_{2\times3} & \multicolumn{2}{c}{I_2}
            \end{array}
        \right], \quad X_T := \left[
            \begin{array}{c;{2pt/2pt}cc}
                Q_T & v_T & p_T \\ \hdashline[2pt/2pt]
                0_{2\times3} & \multicolumn{2}{c}{I_2}
            \end{array}
        \right].
\end{equation}
In this framework, the body and target dynamics \eqref{eq: body target dynamics} can be rewritten as dynamics on $\SE$; see \cite[Lemma 4.1]{vanGoor2025}, in the form
\begin{subequations}\label{eq:XB XT dynamics}
\begin{align}\label{}
    \dot X_B&= X_BU_B+GX_B + NX_B-X_BN, \\
    \dot X_T&= X_TU_T+GX_T + NX_T-X_TN,
\end{align}
\end{subequations}
where
\begin{equation}\label{eq:SE23 U G N}
    \begin{aligned}
    &U_B:=\left[
            \begin{array}{c;{2pt/2pt}cc}
                \Omega_\times & a_B & 0_{3\times 1} \\
                \hdashline[2pt/2pt]
                0_{2\times3} & \multicolumn{2}{c}{0_{2\times2}}
            \end{array}
        \right], \quad U_T:=\left[
            \begin{array}{c;{2pt/2pt}cc}
                \omega_{T\times} & a_T & 0_{3\times1} \\
                \hdashline[2pt/2pt]
                0_{2\times3} & \multicolumn{2}{c}{0_{2\times2}}
            \end{array}
        \right], \\ &G:=\left[
            \begin{array}{c;{2pt/2pt}cc}
                0_{3\times 3} & g & 0_{3\times 1} \\
                \hdashline[2pt/2pt]
                0_{2\times3} & \multicolumn{2}{c}{0_{2\times2}}
            \end{array}
        \right], \qquad\ N:=\left[
            \begin{array}{ccc}
                0_{3\times 3} & 0_{3\times 2} \\
                0_{2\times3} & \begin{bmatrix}
                    0 & -1\\ 0 & 0
                \end{bmatrix}
            \end{array}
        \right].
    \end{aligned}
\end{equation}
Define
\begin{align*}
    R:=Q_T^\top Q_B,\quad\quad
    \xi&:=-Q_B^\top\mathring\xi, \qquad \mathring\xi:=(p_T-p_B),\\
    \nu&:=-Q_B^\top\mathring\nu, \qquad \mathring\nu:=(v_T-v_B),
\end{align*}
where $R\in \SO(3)$ is
the relative orientation from frame $\mathcal{B}$ to frame $\mathcal{T}$, $\mathring\xi\in\R^3$ is the relative position from frame $\mathcal{B}$ to frame $\mathcal{T}$ expressed w.r.t. frame $\mathcal{I}$, $\xi\in\R^3$ is the relative position from frame $\mathcal{B}$ to frame $\mathcal{T}$ expressed w.r.t. frame $\mathcal{B}$, $\mathring\nu\in\R^3$ is the relative velocity from frame $\mathcal{B}$ to frame $\mathcal{T}$ expressed w.r.t. frame $\mathcal{I}$, $\nu\in\R^3$ is the relative velocity from frame $\mathcal{B}$ to frame $\mathcal{T}$ expressed w.r.t. frame $\mathcal{B}$.

Exploiting the $\SE$ group inverse and multiplication, we can express the relative rigid motion from $\mathcal{B}$ to $\mathcal{T}$ as $X\in\SE$ given by

    \begin{align}
    X:=X_B^{-1}X_T&= \left[
            \begin{array}{c;{2pt/2pt}cc}
                Q_B^\top & -Q_B^\top v_B & -Q_B^\top p_B \\
                \hdashline[2pt/2pt]
                0_{2\times3} & \multicolumn{2}{c}{I_2}
            \end{array}
        \right]\left[
            \begin{array}{c;{2pt/2pt}cc}
                Q_T & v_T & p_T \\
                \hdashline[2pt/2pt]
                0_{2\times3} & \multicolumn{2}{c}{I_2}
            \end{array}
        \right]\notag\\&= \left[
            \begin{array}{c;{2pt/2pt}cc}
                Q_B^\top Q_T & Q_B^\top\mathring\nu & Q_B^\top\mathring\xi \\
                \hdashline[2pt/2pt]
                0_{2\times3} & \multicolumn{2}{c}{I_2}
            \end{array}
        \right]  \notag\\
        &=\left[
            \begin{array}{c;{2pt/2pt}cc}
                R^\top & -v & -\xi \\
                \hdashline[2pt/2pt]
                0_{2\times3} & \multicolumn{2}{c}{I_2}
            \end{array}
        \right].\label{eq:X SE23}
    \end{align}  
This notation aligns with the available measurement processes and the estimation problem being addressed.

Differentiating $X$ and substituting \eqref{eq:XB XT dynamics}, it can be verified that the $X$-dynamics are described by
\begin{equation}\label{eq:X SE23 dynamics}
    \dot X=XU_T - U_BX +[N,X].
\end{equation}

Defining the state $(R,\xi,\nu)\in\SO(3) \ltimes \R^6$ from the quantities of interest, we can rewrite \eqref{eq:X SE23 dynamics} as
\begin{equation} \label{eq:relative dynamics}
    \begin{cases}
    \dot \xi &= -\Omega_\times \xi + \nu \\
    \dot{\nu} & = -\Omega_\times \nu + a_B - R^\top a_T\\
    \dot R &= -\omega_{T_\times} R + R\Omega_{\times}.
    \end{cases} 
\end{equation}

As for the output, we first analyze the situation where one has access to the relative position $\xi$ of the origin of the target with respect to the body, expressed in the body frame, which can be obtained, for instance, from a stereo camera:
\begin{equation}\label{eq:output position 1}
    y := \xi \in \mathbb{R}^3.
\end{equation}

We then consider the more challenging case, where only the relative bearing
\begin{equation}\label{eq:out bearing}
    y := \frac{\xi}{|\xi|} \in S^2
\end{equation}
is available as the output measurement. This situation corresponds to either a very distant target or to measurements obtained from a monocular camera.

In both cases, the objective is to design an observer for the system defined by \eqref{eq:relative dynamics}.

\section{Observer Design}\label{section:observer design}

The proposed observer is given by the cascade interconnection of two estimation schemes. The first observer is a \textit{Riccati observer} (\cite{hamel2017position}) developed in the ambient space $\R^{15}$ in order to guarantee linearity of the error dynamics. As a consequence, the estimate of $R$ is not an element of $\SO(3)$. A nonlinear complementary filter (\cite{Mahony_Hamel_Pflimlin}) is then cascaded to the first observer in order to smoothly reconstruct $\SO(3)$ attitude estimates at all times.

\begin{figure}[ht]
    \centering    \includegraphics[width=\linewidth]{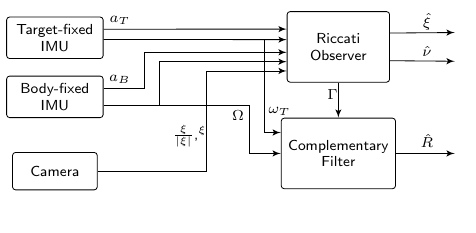}
    \caption{Proposed cascaded observer architecture.}
    \label{fig:enter-label}
\end{figure}

\subsection{Riccati observer in the ambient space $\R^{15}$}
Based on the available measurements and the system dynamics \eqref{eq:relative dynamics}, $Q_B$ and $Q_T$ cannot be uniquely recovered. Only the relative orientation $R$ may be observable.

This follows from the fact that the dynamic model of interest does not depend on $Q_T$, and the measurement equation is likewise independent of $Q_T$, hence $Q_T$ and $Q_B$ can only be estimated up to an unknown constant rotation matrix.

To make this explicit, introduce the auxiliary state variable $Z\in \mathbf{SO}(3)$ that inherits the dynamics of $Q_T$. That is:
\begin{align}\label{eq:Z dynamics}
    \dot Z =Z\omega_{T_\times},\qquad 
    Z(0)=I_3,
\end{align}
and then define the transformed quantities:
\begin{equation*}
      a_T^z = Z a_T,
     \qquad
     R_z = Z R.
\end{equation*}
One verifies that the dynamics \eqref{eq:relative dynamics} can be written as:
\begin{equation} \label{dyn obs}
    \begin{cases}
    \dot \xi &= -\Omega_\times \xi + \nu \\
    \dot{\nu} & = -\Omega_\times \nu + a_B - R_z^\top a^z_T \\
    \dot{R_z} &= R_z \Omega_\times.
    \end{cases} 
\end{equation}
The above model is invariant under arbitrary left translations of $Q_T$. As the measurements do not depend on $Q_T$, the only potentially observable quantity is the relative orientation $R$ or equivalently $R_z$ ($R = Z^\top R_z$).

Let $\hat{\xi}, \hat{\nu} \in \mathbb{R}^3$ denote the estimates of the relative position $\xi$ and relative velocity $\nu$, respectively.  
Define $\gamma_z \in \mathbb{R}^9$ as the estimate of $\mathrm{vec}(R_z^\top)$, so that 
\[
    \Gamma_z := \mathrm{vec}^{-1}(\gamma_z) \in \mathbb{R}^{3 \times 3}
\]
represents the estimate of $R_z^\top$.

The proposed linear observer, with state $(\hat \xi, \hat \nu, \gamma_z)\in\R^{15}$, is of the form
\begin{equation} \label{eq:obs Riccati}
    \begin{cases}
    \dot{\hat{\xi}} &= -\Omega_\times \hat{\xi} + \hat{\nu} + \sigma_\xi \\
    \dot{\hat{\nu}} & = -\Omega_\times \hat{\nu} + a_B - \Gamma_z a^z_T + \sigma_v \\
    \dot{\gamma_z} &= -(I_3\otimes\Omega_\times) \gamma_z+ \sigma_{\gamma},
    \end{cases} 
\end{equation}
with $\sigma_{\xi},\sigma_{\nu}\in \mathbb{R}^3$ and $\sigma_{\gamma}\in\R^9$ the innovation terms to be determined.

\begin{remark}
    By definition, $\Gamma_z$ is not constrained to lie in $\mathrm{SO}(3)$, even though it will be shown to converge to $R_z^\top$. Since $Z$ is well-conditioned, the convergence of $\Gamma_z$ to $R_z^\top$ implies that the estimate $\Gamma := \Gamma_z Z$ converges to $R^\top$ for any $Z \in \mathrm{SO}(3)$ satisfying \eqref{eq:Z dynamics}.
\end{remark}

The linear vectorized dynamics \eqref{eq:obs Riccati} allow to design the observer in the Riccati framework. Consider the error state vector $\bm{x} = (\tilde{\xi}, \tilde{\nu}, \tilde{\gamma}) \in \R^{15}$, with $$ \tilde{\xi} = \xi - \hat{\xi}, \qquad \tilde{\nu} = \nu - \hat{\nu} \qquad \tilde{\gamma}_z := \mathrm{vec}(R_z^\top) - \gamma_z. $$
Then, the error dynamics can be written as
\begin{subequations}\label{dyn x}
    \begin{align}
        \dot{\bm{x}} &= A(t) \bm{x} + \bm{u} \label{eq:x dynamics} \\
        \bm{y} &= C \bm{x}\label{eq:out pos}
    \end{align}
\end{subequations}
with 
\begin{align}
    A(t) &= \bar{A}(t) \otimes I_3 + S(t),\\
    S(t) &= -I_5 \otimes \Omega_\times, \qquad
    \bar{A}(t) = 
    \begin{bmatrix}
        0 & 1 & 0_{1\times 3} \\
        0 & 0 & -a^{z\top}_T \\
        0_{3\times 1} & 0_{3 \times 1} & 0_{3 \times 3}
    \end{bmatrix} \notag \\
\end{align}
The $C$ matrix depends on the considered measurement.

When the position output~\eqref{eq:output position 1} is used, we define
$\mathbf{y} = \xi - \hat{\xi}$, or equivalently:
\begin{equation}
   \bm{y} = C \bm x \mbox{ with } C = \bar C \otimes I_3,
    \qquad
    \bar C = 
    \begin{bmatrix}
        1 & 0_{1\times 4}
    \end{bmatrix}.
    \label{eq:C position}
\end{equation}

When only the bearing measurement~\eqref{eq:output bearing} is used, we exploit the identity
$\pi_y \xi \equiv 0$ and define $\bm{y} = \pi_y(\xi - \hat{\xi})$.  
This yields the time-varying output matrix
\begin{equation}\label{eq:output bearing}
  \bm{y} = C(t) \bm x \mbox{ with }   C(t) = \bar C \otimes \pi_y(t).
\end{equation}

The innovation terms in $\mathbf{u} = [\sigma_p^\top,\; \sigma_\nu^\top,\; \sigma_\gamma^\top]^\top$ are chosen such that $\bm{u} = PC^\top D \bm{y}$, with $P$ solution to the continuous Riccati equation
    \begin{equation}\label{eq:riccati}
        \dot{P} = AP + PA^\top - PC^\top D CP + V, \quad P(0) \in \mathbb{S}_+(15),
    \end{equation}
    with $V \in \mathbb{S}_+(15)$ and $D \in \mathbb{S}_+(3)$ bounded continuous matrix-valued functions.

Uniform observability of the pair $(A,C)$, with $C$ obtained either from
\eqref{eq:C position} or \eqref{eq:output bearing}, together with the positive
definiteness of $D$ and $V$, is sufficient to guarantee well-posedness,
well-conditioning, and boundedness of the solutions to \eqref{eq:riccati} (see
\cite[Lemmas~2.5 and~2.6]{hamel2017position}). This, in turn, guarantees the
global exponential stability of $\bm{x} = 0$.

\begin{lemma}[{\em Position output measurement}]\label{lemma:GES PE a_T}
    Consider the dynamics \eqref{eq:obs Riccati} with output equation \eqref{eq:C position} and let $a_T$ continuous and bounded. If there exist $\delta,\mu>0$ such that for all $t\geq0$
    \begin{equation}\label{eq:PE a}
        \frac{1}{\delta}\int_t^{t+\delta}a^z_T(s) a^{z\top}_T(s)\,ds\succeq \mu I_3, 
    \end{equation}
    then, the pair $(A,C)$ is uniformly observable and the equilibrium point $\bm{x}=0$ is globally exponentially stable.
\end{lemma}
The proof of this result is reported in Appendix~\ref{proof: lemma GES PE a_T}.

Since $a_T^{z} = Z a_T$, the persistence of excitation of $a_T^{z}$ in
\eqref{eq:PE a} results from the combined effect of the target rotational
dynamics, through \eqref{eq:Z dynamics}, and the target linear acceleration
$a_T$. It is straightforward to verify that this condition is fulfilled either
by the persistent excitation of $a_T$ (with a constant target orientation,
i.e., $\omega_T = 0$), or by the persistent excitation of $a_T^{z}$ when $a_T$
is constant.

When position measurements are not directly available and only the bearing measurement \eqref{eq:output bearing} is accessible, stricter conditions are required to guarantee the uniform observability of $(A,C)$. This stems from the inherently limited information provided by bearing-only measurements.

Denote by $y_0 = Q_B y$ the measured bearing expressed in $\mathcal{I}$.  
Define 
\begin{equation}
    M(s,t) := \pi_{y_0}(s)\left(\Lambda_T(s,t) - [\, I_3 \;\; (s-t) I_3 \,]\rho \right),
\end{equation}
where the matrices $\Lambda_T$ and $\rho$ are introduced for a given $\delta>0$:
\begin{align}
    \Lambda_T(s,t) 
        &:= \int_t^s (\tau - s)\, a_T^{z\top}(\tau)\, d\tau \otimes I_3,\\[0.2cm]
     \rho(t,\delta) 
        &:= \Lambda_\pi^{-1}(t,\delta)\, B(t,\delta)
\end{align}    
with    
 \begin{align*}   
    \Lambda_\pi(t,\delta) 
        &:= \int_t^{t+\delta}
            \begin{bmatrix}
                1 & s-t \\
                s-t & (s-t)^2
            \end{bmatrix}
            \otimes \pi_{y_0}(s)\, ds, \mbox{ and }\\[0.2cm]
    B(t,\delta) 
        &:= \int_t^{t+\delta}
            \begin{bmatrix}
                1 \\ s-t
            \end{bmatrix}
            \otimes \pi_{y_0}(s)\, \Lambda_T(s,t)\, ds.
\end{align*}

\begin{lemma}[{\em Bearing output measurement}]\label{lemma:bearing}
    Consider the system \eqref{eq:x dynamics}-\eqref{eq:output bearing} and suppose $a_T$ continuous and bounded. If there exist $\mu,\delta>0$ such that for all $t\geq 0$
    \begin{subequations}
        \begin{align}
            &\frac{1}{\delta}\int_t^{t+\delta}\pi_{y_0}(s)\,ds\succeq\mu I_3 \label{eq:PE pi},\\
            &\frac{1}{\delta}\int_t^{t+\delta} M^\top(s,t) M(s,t)\,ds\succeq\mu I_9,\label{eq:PE schur}
        \end{align}
    \end{subequations}
    then the pair $(A,C)$ is uniformly observable and the equilibrium point $\bm{x}=0$ is globally exponentially stable.
\end{lemma}

The proof of this result is reported in Appendix~\ref{proof:bearing}.

Condition \eqref{eq:PE pi} is standard in bearing-only estimation problems \citep{hamel2017position} and requires (weak) persistence of excitation of $y_0$. Condition \eqref{eq:PE schur}, the Schur complement of the observability Gramian \eqref{eq:x dynamics}--\eqref{eq:output bearing} with respect to the $(\tilde \xi, \tilde \nu)$-subsystem, guarantees uniform observability of the full extended dynamics, including $\gamma_z$, assuming \eqref{eq:PE pi} holds. In this case, uniform observability is more demanding than in the relative position scenario, as it requires joint excitation of the bearing and the target motion. Specifically, there must exist $\delta_1, \mu_1 > 0$ such that, for all $t \geq 0$,
\begin{equation}
    \frac{1}{\delta_1} \int_t^{t+\delta_1} \Lambda_T^\top(s,t) \, \pi_{y_0}(s) \, \Lambda_T(s,t) \, ds \succeq \mu_1 I_9,
\end{equation}
ensuring that $\Lambda_T$ does not align with the bearing, preventing rotations around the bearing axis from becoming unobservable.

\subsection{$\SO(3)$ attitude reconstruction and cascade stability}

Since the goal is to reconstruct an estimate of \eqref{eq:relative dynamics} at all times, an additional step needs to be taken for the attitude reconstruction. A simple solution to guarantee that the attitude estimate is always an element of $\SO(3)$ is to project $\Gamma$ to $\SO(3)$ (\cite{Sarabandi2023}). However, due to the discontinuous nature of the projection operator from the ambient space $\R^{3\times3}$ to $\SO(3)$, the reconstructed rotation matrix can present jumps. A trivial example is if $\Gamma$ crosses the zero matrix.

In this work we, instead, reconstruct an $\SO(3)$ attitude estimate by cascading a nonlinear complementary filter on $\SO(3)$ (\cite{Mahony_Hamel_Pflimlin}) to the proposed observer \eqref{eq:obs Riccati}. This solution provides smooth attitude estimates at all times at the cost of reducing the convergence to the almost-global case due to inherent topological obstructions. The adopted complementary filter takes the form
\begin{equation}\label{eq:comp filter}
    \dot{\hat{R}} = -\omega_{T\times}\hat{R} + \hat R\Omega_\times + \sigma_{\hat{R}\times}\hat{R},
\end{equation}
with $\hat{R}\in\SO(3)$ the $\SO(3)$ estimate of $R$, and $\sigma_{\hat{R}}\in\R^3$ the innovation term defined as
\begin{equation}\label{eq:innov comp filter}
    \sigma_{\hat{R}} := \sum_{i=1}^3k_i(\hat{R}e_i\times \Gamma^\top e_i).
\end{equation}
Let
\begin{align*}
    \mathcal{E}_s &= \{ I_3 \}, &
    \mathcal{E}_u &= \{R \in \SO(3) \mid \mathrm{tr}(R) = -1\},
\end{align*}
denote the two distinct sets of equilibria of \eqref{eq:comp filter}. 
The innovation term \eqref{eq:innov comp filter} guarantees almost-global asymptotic stability of $\hat{R}$ w.r.t. an $\SO(3)$ approximation of $\Gamma^\top$ and, in particular, almost-global asymptotic stability of $R^\top \hat R$ to $\mathcal{E}_s$, with unstable equilibria $\mathcal{E}_u$, when $\Gamma^\top$ converges to $R$.

The following stability result for the cascade interconnection of \eqref{eq:obs Riccati}-\eqref{eq:comp filter} with output \eqref{eq:out pos}, or \eqref{eq:output bearing}, is a direct consequence of \cite[Theorem 6.1]{vanGoor2025}.

\begin{theorem}[{\em Cascade interconnection}]
    Consider the cascade interconnection \eqref{eq:obs Riccati}-\eqref{eq:comp filter}. Under the assumptions of Lemma~\ref{lemma:GES PE a_T} or of Lemma~\ref{lemma:bearing}, the equilibrium $\bm x = 0$ is globally exponentially stable and $R^\top \hat{R}$ is almost globally asymptotically stable and locally exponentially stable w.r.t. $\mathcal{E}_s$, with unstable equilibrium set $\mathcal{E}_u$.
\end{theorem}

\section{Simulations}\label{section:simulations}
This section illustrates the performance of the proposed observer on a scenario inspired by a multirotor UAV (body) hovering above a moving ship (target). The ship follows a circular horizontal trajectory, with orbital angular velocity $\omega_{orb}$, superimposed on a vertical oscillation with angular frequency $\omega_z$ induced by the waves.

The target motion is specified as
\begin{align*}
p_T(0)&=\begin{bmatrix}5 & 0 & 0\end{bmatrix}^\top,\quad 
a_T(t)=Q_T^\top(t)\begin{bmatrix}-5\omega_{orb}^2\cos (\omega_{orb}t) \\ -5\omega_{orb}^2\sin (\omega_{orb}t) \\ -2\omega_z^2\sin(\omega_zt) -g\end{bmatrix},\\
    v_T(0)&=\begin{bmatrix}0 & 5\omega_{orb} & 0.5\omega_z\end{bmatrix}^\top,\quad \omega_T(t) =\begin{bmatrix}
    0& 0 & \omega_{orb}
\end{bmatrix}^\top, 
\end{align*}
while the body motion is characterized by
\begin{align*}
    p_B(t)\equiv \begin{bmatrix}
        0 & 0 & 20
    \end{bmatrix}^\top, \quad  v_B(t) \equiv \Omega(t)\equiv 0,\quad  a_B(t)=-ge_3 .
\end{align*}

The auxiliary state $Z$ is initialized to $Z(0) = I_3$. 

The is implemented with $P(0) = 3I_{15}, \ D = 20I_3$ and $V = 0.001I_{15}$, and initial estimates $\gamma_z(0) =$vec$(I_3),\ \hat \xi = \begin{bmatrix}
    1&-0.5&0.8
\end{bmatrix}^\top,\ \hat \nu = \begin{bmatrix}
    -0.5&0.8&-0.3
\end{bmatrix}^\top$.

The simulation is divided into three phases, and is reported for both output configurations: relative position (Fig.~\ref{fig:relative position plot}) and bearing-only (Fig.~\ref{fig:bearing plot}).

In the initial phase, we set $\omega_{orb} = 0.4$ and $\omega_z = 2$. In this regime, occurring in the white interval before the shaded region in Figs.~\ref{fig:relative position plot}-\ref{fig:bearing plot}, the target undergoes a persistently excited horizontal motion combined with vertical oscillations, and conditions \eqref{eq:PE a}, \eqref{eq:PE pi}, and \eqref{eq:PE schur} are satisfied. As seen in Figs.~\ref{fig:relative position plot}–\ref{fig:bearing plot}, the estimation errors are approaching zero in both the relative position and bearing-only cases.

The middle phase, corresponding to the shaded interval in Figs.~\ref{fig:relative position plot}–\ref{fig:bearing plot}, is characterized by $\omega_{orb} = \omega_z = 0$. The target then follows a constant-velocity trajectory, so that \eqref{eq:PE a} and \eqref{eq:PE schur} are no longer satisfied and uniform observability is lost. This loss of excitation is reflected in the estimation performance: the attitude error ceases to decrease, which in turn causes the position estimate to diverge.

In the final phase, the orbital and vertical angular velocities are restored to $\omega_{orb} = 0.4$ and $\omega_z = 2$. The renewed excitation re-establishes the observability conditions, and the observer drives the full estimation error back to zero. This recovery is clearly visible in both measurement configurations, confirming the link between the proposed PE conditions and the convergence properties of the observer.

\begin{figure}[ht]
    \centering
    \includegraphics[width=\linewidth, trim={0cm 0cm 0cm 0.4cm},clip]{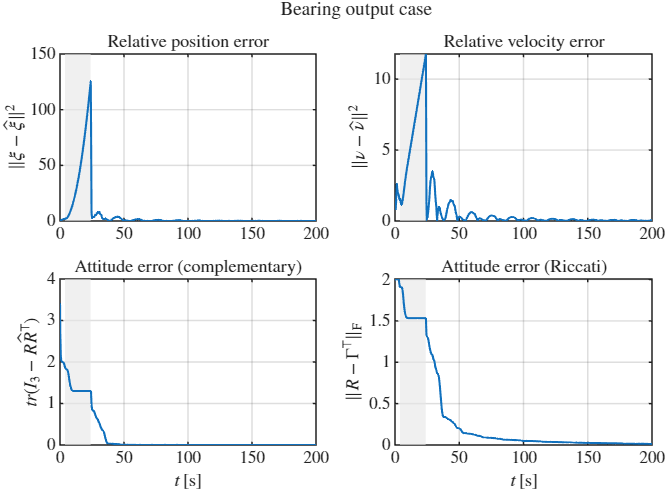}
    \caption{Time evolution of the estimates errors for the relative position output with no acceleration (shaded region) for $t\in[3.927,23.927]$.}
    \label{fig:relative position plot}
\end{figure}
\begin{figure}[ht]
    \centering
    \includegraphics[width=\linewidth, trim={0cm 0cm 0cm 0.4cm},clip]{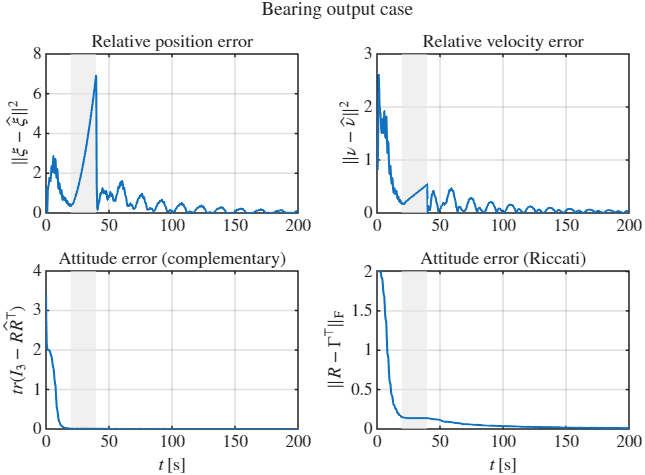}
    \caption{Time evolution of the estimates errors for the bearing-only output with no acceleration (shaded region) for $t\in[19.635,39.635]$.}
    \label{fig:bearing plot}
\end{figure}

\section{Conclusions}\label{section:conclusions}

In this paper, we addressed the estimation of the relative pose and linear velocity of a body with respect to a potentially moving target using a dual-IMU setup, with either relative position or bearing measurements. Rather than designing the observer directly on the Lie group $\SE$, we embedded the extended pose estimates into $\R^{15}$ to express the relative dynamics as a linear time-varying system. This formulation allowed us to design a deterministic Riccati observer that guarantees exponential convergence of the estimation error under uniform observability conditions, derived from the motions of both the target and the vehicle.
To obtain smooth relative attitude estimates, a nonlinear complementary filter on $\SO(3)$ was cascaded with the Riccati observer. This approach yields almost-global asymptotic convergence of the attitude error while preserving exponential convergence of the position and velocity estimates.
Simulation results illustrate the effect of uniform observability on the observer's performance for both relative position and bearing-only measurements.
Future work will consider extending the observer to handle biased IMU measurements.

\section{Appendix}\label{section:appendix}

\subsection{Proof of Lemma~\ref{lemma:GES PE a_T}}\label{proof: lemma GES PE a_T}

Lemma~\ref{lemma:obs Abar,C = obs A,C}, reported in Appendix~\ref{section:Technical Lemmas}, guarantees that uniform observability of
$(A,C)$ is equivalent to uniform observability of $(\bar{A}\otimes I_3,TCT^\top)$, where $T:=I_5\otimes Q_B$.

By definition \eqref{eq:C position}, it can be verified that $C= TCT^\top$. Moreover, from a similar argument to the proof of Lemma 1 in \cite[Appendix A]{Benahmed2025}, we can conclude that uniform observability of the pair $(\bar{A}\otimes I_5, C)$ is equivalent to uniform observability of the pair $(\bar{A},\bar{C})$. Therefore, we can restrict the following analysis to uniform observability of $(\bar{A},\bar{C})$.

The state transition matrix $\bar\Phi(s,t)$ associated with $\bar{A}(t)$ can be computed from the Peano-Baker series expansion, yielding
\begin{align*}
\bar\Phi(s,t) &= 
\begin{bmatrix}
     1 & s-t & -\int_{t}^{s}\int_{t}^{\tau}a^{z\top}_T(\sigma)\, d\sigma\, d\tau \\
     0 & 1 & -\int_{t}^{s}a^{z\top}_T(\tau)\,d\tau\\
     0_{3\times 1} & 0_{3 \times 1} & 0_{3\times 3}
 \end{bmatrix}
 .
\end{align*}
It follows that
\begin{align*}
    \bar C\bar{\Phi}(s,t) &= 
    \left[1 \; 0_{1\times 4}  \right]\begin{bmatrix}
        1 & s-t & -\int_{t}^{s}\int_{t}^{\tau}a^{z\top}_T(\sigma)\, d\sigma\, d\tau \\
        0 & 1 & -\int_{t}^{s}a^{z\top}_T(\tau)d\tau\\
        0_{3\times 1} & 0_{3 \times 1} & I_3
    \end{bmatrix}  \\
    &= \begin{bmatrix}
        1 & s-t & -\int_{t}^{s}\int_{t}^{\tau}a^{z\top}_T(\sigma)\, d\sigma\, d\tau
    \end{bmatrix}.
\end{align*}
Denote $b(s,t):=-\int_{t}^{s}\int_{t}^{\tau}a^z_T(\sigma)\, d\sigma\, d\tau$.
The observability Gramian $\overline{W}(t, t+ \delta)$ takes the form
\begin{align*}
    \overline{W}(t, t+ \delta) 
    &=  \frac{1}{\delta} \int_t^{t+\delta} \bar{\Phi}^\top(s,t) \bar{C}^\top \bar{C}\bar{\Phi}(s,t) ds \\
    &= \frac{1}{\delta} \int_t^{t+\delta}\begin{bmatrix}
        1 \\ s-t \\ b(s,t)
    \end{bmatrix} 
    \begin{bmatrix}
        1 & s-t & 
        b^\top(s,t)
    \end{bmatrix}
    \, ds\\
    &=\frac{1}{\delta}
    \begin{bmatrix}
        \delta  & \frac{\delta^2}{2} & \int_t^{t+\delta}b^\top \\
        \frac{\delta^2}{2} & \frac{\delta^3}{3} & \int_t^{t+\delta}(s-t)b^\top \\
        \int_t^{t+\delta}b & \int_t^{t+\delta}(s-t)b & \int_t^{t+\delta}bb^\top
    \end{bmatrix}.
\end{align*}

Rewrite $\overline{W}$ as the block matrix
    \begin{align*}
        \overline{W}=\frac{1}{\delta}\begin{bmatrix}
            \Lambda_\delta & B \\
            B^\top & \Lambda_b
        \end{bmatrix}
        \end{align*}
        with
        \begin{align*}
        &\Lambda_\delta = 
        \begin{bmatrix}
             \delta & \frac{\delta^2}{2} \\ 
             \frac{\delta^2}{2} & \frac{\delta^3}{3}
        \end{bmatrix},
        \quad B = 
        \begin{bmatrix}
        \int_t^{t+\delta}b^\top(s,t)\,ds \\
        \int_t^{t+\delta}(s-t)b^\top(s,t)\,ds
        \end{bmatrix},\\
        &\Lambda_b = \begin{bmatrix}
            \int_t^{t+\delta}b(s,t)b^\top(s,t)\,ds
        \end{bmatrix} .
    \end{align*}
It is straightforward to prove that $\Lambda_\delta$ is uniformly positive definite for all $\delta>0$.
Therefore, by the Lemma~\ref{lemma:schur} in Appendix~\ref{section:Technical Lemmas}, it is sufficient to prove uniform positive definiteness of the Schur complement of the Gramian w.r.t. $\Lambda_{\delta}$, i.e. of $\overline{W}/\Lambda_\delta := \Lambda_b - B^\top \Lambda_\delta^{-1} B$, in order to prove uniform positive definiteness of the Gramian.

Denote $d(s,t):= \left[1\ \ (s-t)\right]^\top\in\R^2$, and notice that
\begin{align*}
\int_t^{t+\delta}d(s,t)d^\top(s,t)\,ds  = \Lambda_\delta,\quad
\int_t^{t+\delta}d(s,t)b^\top(s,t)\,ds = B.
\end{align*}
Direct computation shows that 
\begin{align*}
    \overline{W}/\Lambda_\delta = \frac{1}{\delta}\int_{t}^{t+\delta} 
\left(b^\top - d^\top \Lambda_\delta^{-1} B\right)^\top\left(b^\top - d^\top \Lambda_\delta^{-1} B\right)\, ds.
\end{align*}
Condition \eqref{eq:schur_complement_lemma3} implies that there exist $\delta,\mu>0$ such that for all $t\geq0$, $v\in S^2$,
\begin{equation*}\label{eq:Schur obs}
    v^\top \overline{W}/\Lambda_\delta v 
= \frac{1}{\delta}\int_{t}^{t+\delta} 
\norm{ \left(b^\top(s,t) - d^\top(s,t) \Lambda_\delta^{-1} B\right) v }^2\, ds \geq \mu.
\end{equation*}
Following \cite{Morin2017}, let us proceed by contradiction, assuming the above condition false; that is, for all $\delta,\mu>0$ there exist a sequence $\{t_k\}_{k\in \mathbb{N}}$, and $v\in S^2$ such that
\begin{equation}\label{eq:integral limit CRv}
    \lim_{k\to\infty} \frac{1}{\delta} \int_{t_k}^{t_k+\delta} \norm{\left(b^\top(s,t) - d^\top(s,t) \Lambda_\delta^{-1} B\right) v}^2ds = 0.
\end{equation}
By continuity and boundedness of $a^z_T$, it follows that $\frac{d^2}{ds^2}\left(b^\top(s,t) - d^\top(s,t) \Lambda_\delta^{-1} B\right) v$ is well defined and bounded. Applying Lemma 1 in \cite{Morin2017} to \eqref{eq:integral limit CRv}, we obtain
\begin{align*}
\lim_{k\to\infty} \frac{1}{\delta}\int_{t_k}^{t_k+\delta}\norm{\frac{d^2}{ds^2}\left(b^\top(s,t_k) - d^\top(s,t_k) \Lambda_\delta^{-1} B\right) v}^2ds = 0. \label{eq:dds integrand Schur}
\end{align*}
By direct differentiation, it follows that
\begin{align*}
    \frac{d^2}{ds^2}\left(b^\top(s,t_k) - d^\top(s,t_k) \Lambda_\delta^{-1} B\right) v= -a^{z\top}_T(s)v.
\end{align*}
This result implies
\begin{align*}
\lim_{k\to\infty}\frac{1}{\delta}\int_{t_k}^{t_k+\delta}v^\top a^z_T(s)a^{z\top}_T(s)v\,ds =0,
\end{align*}
contradicting \eqref{eq:PE a} and proving the lemma.

\subsection{Proof of Lemma~\ref{lemma:bearing}}\label{proof:bearing}
By Lemma~\ref{lemma:obs Abar,C = obs A,C} in Appendix~\ref{section:Technical Lemmas}, uniform observability of the pair $(A,C)$ is equivalent to uniform observability of the pair $(\bar{A}(t)\otimes I_5,T(t)C(t)T^\top(t))$, where $T(t):=I_5\otimes Q_B(t)$. The state transition matrix $\Phi(s,t)$ associated with $\bar A(t)\otimes I_5$ can be expressed as
\begin{align*}
\Phi(s,t) &= 
\begin{bmatrix}
     1 & s-t & -\int_{t}^{s}\int_{t}^{\tau}a^{z\top}_T(\sigma)\, d\sigma\, d\tau \\
     0 & 1 & -\int_{t}^{s}a^{z\top}_T(\tau)\,d\tau\\
     0_{3\times 1} & 0_{3 \times 1} & I_3
 \end{bmatrix}\otimes I_3.
\end{align*}
From \eqref{eq:output bearing}, it follows that $C_T(t):=T(t)C(t)T^\top(t) = [1\ \ 0_{1\times 4}]\otimes \pi_{y_0}$. As a consequence, the term $C_T(s)\Phi(s,t)$ takes the form
\begin{align*}
     &C_T(s)\Phi(s,t) =\\&= 
    (\left[1 \; 0_{1\times 4}  \right]\otimes \pi_{y_0})\begin{bmatrix}
        1 & s-t & -\int_{t}^{s}\int_{t}^{\tau}a^{z\top}_T(\sigma)\, d\sigma\, d\tau \\
        0 & 1 & -\int_{t}^{s}a^{z\top}_T(\tau)d\tau\\
        0_{3\times 1} & 0_{3 \times 1} & I_3
    \end{bmatrix}\otimes I_5  \\
    &= \begin{bmatrix}
        1 & s-t & -\int_{t}^{s}\int_{t}^{\tau}a^{z\top}_T(\sigma)\, d\sigma\, d\tau
    \end{bmatrix}\otimes \pi_{y_0}\\
    &= \pi_{y_0}\begin{bmatrix}
        I_3 & (s-t)I_3 & \Lambda_T(s,t)
    \end{bmatrix} ,
\end{align*}
where we used the identity $(A\otimes B)(C\otimes D)= (AC)\otimes (BD)$ for $A$, $B$, $C$, $D$ matrices of appropriate dimensions.

The Gramian $W(t, t+ \delta)$ associated to the pair $(\bar{A}(t)\otimes I_5,C_T(t))$ takes the form
\begin{align*}
    &W(t, t+ \delta) =\\
    &=  \frac{1}{\delta} \int_t^{t+\delta} \Phi^\top(s,t) C_T(s)^\top C_T(s)\Phi(s,t) ds \\
    &=  \frac{1}{\delta} \int_t^{t+\delta} \begin{bmatrix}
        I_3 & (s-t)I_3 & \Lambda_T
    \end{bmatrix}^\top\pi_{y_0}\begin{bmatrix}
        I_3 & (s-t)I_3 & \Lambda_T
    \end{bmatrix} ds \\
    &= \frac{1}{\delta} \int_t^{t+\delta} \begin{bmatrix}
        \pi_{y_0} & (s-t)\pi_{y_0} & \pi_{y_0}\Lambda_T \\
        (s-t)\pi_{y_0} & (s-t)^2\pi_{y_0} & (s-t)\pi_{y_0}\Lambda_T\\
        \Lambda_T^\top \pi_{y_0} & \Lambda_T^\top (s-t)\pi_{y_0} & \Lambda_T^\top \pi_{y_0}\Lambda_T
    \end{bmatrix} ds.
\end{align*}

Notice that $W$ can be rewritten as the block matrix
    \begin{align}
        W=\frac{1}{\delta}\begin{bmatrix}
            A_\pi & B \\
            B^\top & A_\Lambda
        \end{bmatrix}
        \end{align}
        where
        \begin{align}
        A_\Lambda = \int_t^{t+\delta}\Lambda_T^\top(s,t) \pi_{y_0}(s)\Lambda_T(s,t)ds.
    \end{align}
By \cite[Lemma 2.7]{hamel2017position}, condition \eqref{eq:PE pi} is sufficient to guarantee that $A_\pi$ is uniformly positive definite. 
Moreover, it can be shown by direct computation that \eqref{eq:PE schur} is equivalent to uniform positive definiteness of the $W/A_\pi$, proving the claim by Lemma~\ref{lemma:schur} in Appendix~\ref{section:Technical Lemmas}.

\subsection{Technical Lemmas}\label{section:Technical Lemmas}
\begin{lemma}\label{lemma:obs Abar,C = obs A,C}
    Consider system \eqref{dyn x}. Denote $T(t):=I_5\otimes Q_B(t)$, then uniform observability of the pair $(A,C)$ is equivalent to uniform observability of the pair $(\bar{A}\otimes I_3,TCT^\top)$.
\end{lemma}
\begin{proof}
Consider the change of coordinates $\bar{\bm x}(t) = T(t) \bm x(t)$. 
Then, by direct differentiation,
\begin{align*}
    \dot{\bar{\bm{x}}} =& \dot T  \bm{x} + T \dot{\bm{x}}  \\
    =& (I_5\otimes (Q_B \Omega_\times)) \bm{x} + (I_5 \otimes Q_B)(\bar{A} \otimes I_3 -I_5\otimes \Omega_\times)\bm{x} \\
    =& (I_5\otimes (Q_B \Omega_\times)) \bm{x} + (I_5 \otimes Q_B)(\bar{A} \otimes I_3)\bm{x} \\&- (I_5 \otimes Q_B) (I_5\otimes\Omega_\times) \bm{x} \\
    =& (I_5 \otimes Q_B)(\bar{A} \otimes I_3)\bm{x} \\
    =& (\bar{A} \otimes I_3)(I_5 \otimes Q_B)\bm{x} \\
    =& (\bar{A} \otimes I_3)\bar{\bm{x}},
\end{align*}
where we used the identity $(A\otimes B)(C\otimes D)= (AC)\otimes (BD)$ for $A$, $B$, $C$, $D$ matrices of appropriate dimensions.
The state transition matrix $\Phi(s,t)$ of \eqref{dyn x} is related to the state transition matrix $\bar{\Phi}(s,t)$ of the $\dot{\bar{\bm x}}$-dynamics by
\begin{align*}
    \Phi(s,t) = T^\top(s) \bar{\Phi}(s,t)T(t).
\end{align*}
The observability Gramian of $(A(t), C(t))$ is
\begin{align*}
    &W(t, t+ \delta) =\\&= \frac{1}{\delta} \int_t^{t+\delta} \Phi^\top(s,t) C^\top(s)C(s)\Phi(s,t) ds\\
    &= \frac{1}{\delta} \int_t^{t+\delta} T^\top (t)\bar{\Phi}^\top(s,t) T(s) C^\top(s)C(s)T^\top(s) \bar{\Phi}(s,t)T(t) ds\\
    &= T^\top(t) \frac{1}{\delta} \int_t^{t+\delta}  \bar{\Phi}^\top(s,t) C_T^\top(s)C_T(s) \bar{\Phi}(s,t) ds\ T(t),
\end{align*}
where $C_T(t) = T(t)C(t)T^\top(t)$. 
Hence uniform observability of the pair $(\bar{A}(t)\otimes I_3, T(t)C(t)T^\top(t))$ implies uniform observability of the pair $(A(t), C(t))$.    
\end{proof}
 \begin{lemma}\label{lemma:schur}
        Let
        \begin{align*}
            W(t,t+\delta)&=\begin{bmatrix}
                W_E(t,\delta)& W_F(t,\delta)\\
                W_F^{\top}(t,\delta) & W_G(t,\delta)
            \end{bmatrix}\\&:=\frac{1}{\delta}\int_t^{t+\delta} \begin{bmatrix}
            E(s) & F(s) \\ F^\top(s) & G(s)
            \end{bmatrix}ds,
        \end{align*}
        with $E\in\R^{n\times n}$, $F \in \R^{n\times m}$ and $G \in \R^{m\times m}$ continuous uniformly bounded matrix-valued functions.
        
        Suppose there exist $\delta,\mu>0$ such that, for all $t\geq 0$,
        \begin{subequations}\label{eq:schur_lemma_conditions}
        \begin{align}
            &W_E(t,\delta) \succeq\mu I_n,\\
            &W/W_E(t,\delta):=\label{eq:schur_complement_lemma3} W_G-W_F^\top W_E^{-1}W_F \succeq\mu I_m.
        \end{align}        \end{subequations} 
        Then, there exist $\delta,\mu^\star>0$ such that, for all $t\geq0$,
        \begin{equation}\label{eq:result_lemma3}
            W(t,t+\delta)\succeq\mu^\star I_3.
        \end{equation}
    \end{lemma}
\begin{proof}
Conditions \eqref{eq:schur_lemma_conditions} imply the existence of $\delta,\mu>0$ such that for all $t\geq0$,
    \begin{align*}\lambda_{min}\left(W_E(t,\delta)\right)\geq\mu, \quad\lambda_{min}\left(W/W_E(t,\delta)\right)\geq\mu,
    \end{align*}
    where $\lambda_{min}(N)$ (resp. $\lambda_{max}(N)$) represents the minimum (resp. maximum) eigenvalue of the square matrix $N\in\R^{p\times p}$, for some $p\in\mathbb{N}$. This, in turn, implies
    \begin{equation*}
        \det(W_E)\geq \mu^{n},\quad \det(W/W_E) \geq \mu^{m}.
    \end{equation*}

    The Schur complement satisfies the determinant property
    \begin{equation*}\label{eq:Schur_determinant}
        \text{det}(W) = \text{det}(W_E)\text{det}(W/W_E)\geq \mu^{n+m}.
    \end{equation*}

    Boundedness of the matrices $E,F$ and $G$ guarantees, by the Gershgorin circle theorem, the existence of $\mu_{max}>0$ such that $\lambda_{max}\left(W(t,t+\delta)\right)\leq\mu_{max}$,
    which implies
    \begin{equation*}
        \det(W) \leq\lambda_{min}(W)\mu_{max}^{m+n-1}.
    \end{equation*}
    Combining the last two inequalities yields
    \begin{equation*}
        \mu^{m+n}\leq \det(W)\leq \lambda_{min}\mu_{max}^{m+n-1}\Rightarrow \lambda_{min}(W)\geq \frac{\mu^{n+m}}{\mu_{max}^{n+m-1}},
    \end{equation*}
    which is equivalent to condition \eqref{eq:result_lemma3} with $\mu^\star = \frac{\mu^{n+m}}{\mu_{max}^{n+m-1}}$.
\end{proof}

\bibliography{ifacconf}

\end{document}